\documentclass[sigconf]{acmart}

\usepackage{amsmath,amsfonts,amssymb}
\usepackage[ruled,vlined]{algorithm2e}
\usepackage{xcolor}
\usepackage[binary-units]{siunitx}
\usepackage{numprint}
\usepackage{pgfplots}\pgfplotsset{compat=1.14}
\usepackage{booktabs}
\usepackage{balance}
\usepackage[disable]{todonotes}

\newcommand{\set}[1]{\{#1\}}

\newcommand{\TS}{\mathcal{TS}}
\newcommand{\ES}{\mathcal{ES}}
\newcommand{\SH}{\mathcal{SH}}

\newcommand{\calA}{\mathcal{A}}
\newcommand{\PROP}{\ensuremath{\mathsf{PROPYLA}}\xspace}

\newcommand{\oram}{\ensuremath{\texttt{ORAM}}\xspace}
\newcommand{\SHARE}{\ensuremath{\texttt{SHARE}}\xspace}

\newcommand{\calK}{\mathcal{K}}
\newcommand{\pluseq}{\mathrel{+}=}

\newcommand{\OpRead}{\text{`Read'}}
\newcommand{\OpWrite}{\text{`Write'}}
\newcommand{\OpReCom}{\text{`ReCom'}}
\newcommand{\OpReTs}{\text{`ReTs'}}
\newcommand{\Write}{\mathsf{Write}}
\newcommand{\Read}{\mathsf{Read}}

\newcommand{\VerCom}{\mathsf{VerCom}}
\newcommand{\VerTs}{\mathsf{VerTs}}
\newcommand{\TA}{\mathsf{TA}}

\newcommand{\Setup}{\mathsf{Setup}}
\newcommand{\RenewCom}{\mathsf{RenewCom}}
\newcommand{\Commit}{\mathsf{Commit}}

\newcommand{\Share}{\mathsf{Share}}
\newcommand{\Reshare}{\mathsf{Reshare}}
\newcommand{\Reconstruct}{\mathsf{Reconstruct}}

\newcommand{\RenewTs}{\mathsf{RenewTs}}

\newcommand{\VerInt}{\mathsf{VerInt}}

\let\oldnl\nl
\newcommand{\nonl}{\renewcommand{\nl}{\let\nl\oldnl}}%Remove line number for one line
\SetKwData{IN}{\textbf{in}}
\SetKwData{EVERY}{\textbf{every}}
\SetKwData{IS}{\textbf{is}}
\SetKwData{TO}{\textbf{to}}
\SetKwData{DEF}{\textbf{Define}}
\newcommand{\Exp}[2]{\mathbf{Exp}^\mathrm{#1}_{#2}}
\newcommand{\V}{\ensuremath{\mathsf{VIEW}}\xspace}

\newcommand{\op}{\ensuremath{\mathsf{op}}\xspace}
\newcommand{\dat}{\ensuremath{\mathsf{dat}}\xspace}
\newcommand{\id}{\ensuremath{\mathsf{id}}\xspace}

\newcommand{\tnrc}[1]{t_{\mathrm{NRC}(#1)}}

\newcommand{\Clock}{\mathsf{Clock}}
\newcommand{\Client}{\mathsf{Client}}

\newcommand{\AP}{P}
\newcommand{\Loc}{T}
\newcommand{\GetId}{\mathsf{GetId}}
\newcommand{\GenAP}{\mathsf{GenAP}}
\newcommand{\CSI}{\mathsf{CSI}}
\newcommand{\thetime}{\mathsf{time}}
\newcommand{\calS}{\mathcal{S}}

\newcommand{\LINCOS}{\ensuremath{\mathsf{LINCOS}}\xspace}
\newcommand{\Access}{\ensuremath{\mathsf{Access}}\xspace}

\newcommand{\ver}{\mathsf{ver}}

\newcommand{\Stamp}{\mathsf{Stamp}}

\newcommand{\ConfSys}{\mathcal{CS}}
\newcommand{\IntSys}{\mathcal{IS}}

\newcommand{\tst}{\mathsf{ts}}
\begin{document}

%END HEAD

%BEGIN TITLE

\title{\PROP: Privacy-Preserving Long-Term Secure Storage}
\subtitle{(Revised Version, April 2019)}
\thanks{This is the revised version of a paper that appeared in the proceedings of the 6th International Workshop on Security in Cloud Computing (SCC'18) co-located with the 13th ACM ASIA Conference on Computer and Communications Security (ASIACCS'18). See the appendix for a version history.}

%!TEX root = ../main.tex

\settopmatter{printacmref=false}
\makeatletter
\renewcommand\@formatdoi[1]{\ignorespaces}
\makeatother
\renewcommand\footnotetextcopyrightpermission[1]{}

\fancyhead{}
\pagestyle{plain}

\if01
\begin{CCSXML}
<ccs2012>
<concept>
<concept_id>10002978</concept_id>
<concept_desc>Security and privacy</concept_desc>
<concept_significance>500</concept_significance>
</concept>
<concept>
<concept_id>10002978.10002979.10002981</concept_id>
<concept_desc>Security and privacy~Public key (asymmetric) techniques</concept_desc>
<concept_significance>500</concept_significance>
</concept>
<concept>
<concept_id>10002978.10002979.10002984</concept_id>
<concept_desc>Security and privacy~Information-theoretic techniques</concept_desc>
<concept_significance>500</concept_significance>
</concept>
<concept>
<concept_id>10002978.10003018</concept_id>
<concept_desc>Security and privacy~Database and storage security</concept_desc>
<concept_significance>500</concept_significance>
</concept>
</ccs2012>
\end{CCSXML}

\ccsdesc[500]{Security and privacy}
\ccsdesc[500]{Security and privacy~Public key (asymmetric) techniques}
\ccsdesc[500]{Security and privacy~Information-theoretic techniques}
\ccsdesc[500]{Security and privacy~Database and storage security}

\keywords{Long-term secure storage, Renewable cryptography, Information-theoretic cryptography, Privacy, Integrity, Confidentiality.}

\copyrightyear{2018} 
\acmYear{2018} 
\setcopyright{acmlicensed}
\acmConference[SCC'18]{6th International Workshop on Security in Cloud Computing}{June 4, 2018}{Incheon, Republic of Korea}
\acmBooktitle{SCC'18: 6th International Workshop on Security in Cloud Computing, June 4, 2018, Incheon, Republic of Korea}
\acmPrice{15.00}
\acmDOI{10.1145/3201595.3201599}
\acmISBN{978-1-4503-5759-3/18/06}
\fi

\author{Matthias Geihs, Nikolaos Karvelas, Stefan Katzenbeisser, and Johannes Buchmann}
\affiliation{\institution{Technische Universit\"at Darmstadt, Germany}}

%!TEX root = ../main.tex

\begin{abstract}
An increasing amount of sensitive information today is stored electronically and a substantial part of this information (e.g., health records, tax data, legal documents) must be retained over long time periods (e.g., several decades or even centuries). When sensitive data is stored, then integrity and confidentiality must be protected to ensure reliability and privacy.
Commonly used cryptographic schemes, however, are not designed for protecting data over such long time periods.
%Cryptographic schemes allow to ensure these protection goals, but the commonly used schemes
%are not designed for protecting data over such long time periods as their security degrades over time due to advancing computational technology (e.g., today cell phones can break DES and in the future quantum computers may break RSA).
%
Recently, the first storage architecture combining long-term integrity with long-term confidentiality protection was proposed (\mbox{AsiaCCS'17}).
%Recently, the first storage architecture was proposed that provides integrity and confidentiality protection over long time periods by combining information-theoretic cryptography with renewable cryptography (\mbox{AsiaCCS'17}).
However, the architecture only deals with a simplified storage scenario where parts of the stored data cannot be accessed and verified individually. If this is allowed, however, not only the data content itself, but also the access pattern to the data (i.e., the information which data items are accessed at which times) may be sensitive information.

Here we present the first long-term secure storage architecture that provides long-term access pattern hiding security in addition to long-term integrity and long-term confidentiality protection. %and at the same time ensures that the access pattern does not leak any sensitive information to the storage servers.
To achieve this, we combine information-theoretic secret sharing, renewable timestamps, and renewable commitments with an information-theoretic oblivious random access machine.
Our performance analysis of the proposed architecture shows that achieving long-term integrity, confidentiality, and access pattern hiding security is feasible.
\end{abstract}

\maketitle

%END TITLE

%BEGIN CONTENT

%!TEX root = ../main.tex

\section{Introduction}

\subsection{Motivation and Problem Statement}
Large amounts of sensitive data (e.g., health records, governmental documents, enterprise documents, land registries, tax declarations) are stored in cloud-based data centers and require integrity and confidentiality protection over long time periods (e.g., several decades or even centuries).
%For example, Japanese hospitals plan to store medical data in a distributed cloud storage architecture \cite{kuroda2012simulating}.
%\cite{kuroda2012simulating,kuroda2013applying}
%\MG{need 2nd example}
%
Here, by \emph{integrity} we mean that illegitimate changes to the data can be discovered and by \emph{confidentiality} we mean that only authorized parties can access the data.

%\subsubsection{Short-Term Protection.}
Typically, cryptographic signature and encryption schemes (e.g., RSA \cite{Rivest:1978:MOD:359340.359342} and AES \cite{AES}) are used to ensure data integrity and confidentiality.
However, the security of currently used cryptographic schemes relies on computational assumptions (e.g., that the prime factors of a large integer cannot be computed efficiently).
Such schemes are called \emph{computationally secure}.
Due to technological progress, however, computers are becoming more powerful over time and computational assumptions made today are likely to break at some point in the future (e.g., factoring large integers will be possible using quantum computers).
Thus, commonly used cryptographic schemes, which rely on a single computational assumption, are insufficient for protecting data over long periods of time (e.g., decades or even centuries).

%\subsubsection{Long-Term Protection.}
To enable long-term integrity protection, Bayer et al.~\cite{Bayer1993} proposed a method for renewing digital signatures by using timestamps.
%The idea is to cryptographically timestamp a signed data object together with its signature. % at a time when the corresponding signature scheme is considered secure.
%Then, as long as the used timestamp scheme is considered secure, the timestamp proves that the message-signature pair existed when it was valid, and therefore it is still considered valid even if the signature scheme has become insecure.
Following their approach, several other schemes had been developed of which \cite{DBLP:journals/compsec/VigilBCWW15} gives an overview.
It should be noted that only recently the concrete security of these schemes has been investigated and proven \cite{geihs2016security,buldas2017pra,Buldas2017}.
%The modular protection scheme \textsf{Mops} proposed by Weinert et al.\ \cite{DBLP:conf/ccs/WeinertDVGB17} combines these solutions into a modular protection scheme such that for a given application scenario the most efficient scheme can be selected.
%
While these results indicate that integrity protection can be renewed using timestamps, it appears entirely infeasible to prolong confidentiality.
The reason is that once a ciphertext generated by a computationally secure encryption scheme is observed by an attacker, the attacker can always break the encryption using brute-force and a sufficiently large amount of computational resources over time.
While honey encryption \cite{juels2014honey} can help against brute-force attacks in some scenarios, the only rigorous way to provide long-term confidentiality protection is to use \emph{information-theoretically secure} schemes, as their security does not rely on computational assumptions.
Information-theoretic solutions exist for many cryptographic tasks. For example, secure data storage can be realized using secret sharing \cite{Shamir:1979:SS:359168.359176,Herzberg1995} and secure communication can be realized by combining Quantum Key Distribution \cite{gisin2002quantum} with One-Time-Pad Encryption \cite{shannon1949communication}. We refer the reader to \cite{Braun2014} for an overview of techniques relevant to long-term confidentiality protection.
Recently, the first storage system that combines long-term integrity with long-term confidentiality protection has been proposed by Braun et al.~\cite{Braun:2017:LSS:3052973.3053043}.
They achieve this by combining secret sharing with renewable timestamps and unconditionally hiding commitments.
However, their scheme does not allow for accessing and verifying parts of the stored data separately.

Here we consider a setting where the database consists of a list of blocks and each data block can be retrieved and verified separately.
In this setting, however, not only the data content may be regarded as sensitive information, but also the access pattern to the data (i.e., which data blocks are accessed at which times).
%This is particularly true if additional information about the structure of the data is known.
For example, if genome data is stored and accessed at a subset of locations $S$ that are known to be associated with a certain property $X$, then it is likely that the person accessing the data is concerned with $X$.
In such a scenario it is desirable not only to ensure the confidentiality of the data content, but also to ensure the confidentiality of the data access pattern.

\subsection{Contribution}

In this paper we solve the problem of long-term secure data storage with access pattern hiding security by presenting the storage architecture $\PROP$.

Long-term integrity and long-term confidentiality protection in \PROP are achieved similar as in \cite{Braun:2017:LSS:3052973.3053043}.
%For confidentiality-preserving long-term integrity protection, renewable unconditionally hiding commitments are combined with renewable timestamps.
%If necessary, commitments and timestamps can be renewed based on the approaches described in \cite{Bayer1993,geihs2016security,buldas2017pra,Buldas2017}.
%The confidentiality of the data with respect to the timestamp servers is achieved by first committing to the data and then having the commitments timestamped instead of the original data. This technique to preserve long-term confidentiality of the data was first proposed by Braun et al.\ \cite{Braun:2017:LSS:3052973.3053043}.
%\PROP supports the renewal of commitments and timestamps in case their security is threatened (e.g., if a public key certificate is about to expire or a new cryptographic attack has been discovered).
Confidential data is stored at a set of shareholders using information-theoretic secure secret sharing and integrity protection is achieved by combining renewable commitments with renewable timestamps.
%Integrity protection is maintained help of an evidence service that renews integrity protection if necessary.
An evidence service helps the user with maintaining integrity protection and renewing the protection when necessary.
%The usage of unconditionally hiding commitments ensures that the evidence service does not learn the content of the stored data.
%Confidentiality protection of the stored data is preserved by using unconditionally hiding commitments which ensures that the timestamp servers, which issue the timestamps, do not learn anything about the timestamped data.
%A set of shareholders is responsible for storing the confidential data using secret sharing.

We now explain how we additionally achieve long-term access pattern hiding security.
The main idea is to integrate an information theoretically secure oblivious random access machine (ORAM) \cite{Goldreich:1987:TTS:28395.28416} with the shareholders and the evidence service so that none of these storage servers learns the access pattern of the client.
%We stress that our solution provides information theoretic confidentiality and access pattern hiding which means that these properties are ensured even against computationally unbounded adversaries.
%Also, it ensures renewable integrity protection which means that as soon as the security of one of the used commitment or timestamp schemes is threatened, it can be replaced by a new scheme without sacrificing the integrity protection.
%
The challenge here is to ensure that even under composition the individual security properties (i.e., renewable integrity protection, information-theoretic confidentiality, and information-theoretic access pattern hiding security) are preserved.
To preserve information-theoretic access pattern hiding security, the data flows induced by any two different database accesses must be indistinguishable from the viewpoint of any of the storage servers.
Typically, this is solved by using an ORAM in combination with shuffling and re-encrypting the accessed data blocks.
This approach, however, cannot be directly applied in our case, as here commitments are stored at the evidence service which must be timestamped and therefore cannot be stored encrypted.
We solve the problem of obfuscating the reshuffling by employing a recommitment technique that enables us to refresh commitments while maintaining their binding property.
%Thereby we achieve indistinguishability of the access patterns on the evidence server side.
%
With regard to the secret shareholders it is sufficient to regenerate the secret shares on every data access in order to obfuscate the data block shuffling.
This ensures that there is no correlation between the received and stored data blocks, and thus the shuffling of the data blocks cannot be traced.
%Thereby we achieve indistinguishability of the access patterns on the storage server side.
%
We also analyze the security of \PROP and show it indeed achieves long-term integrity, long-term confidentiality, and long-term access pattern hiding security.
Finally, we evaluate the performance of \PROP and show that storage, computation, and communication costs appear manageable.
%In particular, we show that in comparison to a system that combination of secret sharing and ORAM that does not achieve long-term integrity protection we show that the storage overhead per block is independent of the block size and the computation and communication overhead is logarithmic in the database size while also depending on how much timestamp and commitment data has already been accumulated.
%The timestamp and commitment renewal time scale linearly with the database size.
%We remark that using Secret Sharing introduces a linear communication overhead in the number of share holders. However, this overhead seems inevitable when long-term confidentiality protection is required, as secret sharing is the only solution for information-theoretic confidentiality protection of data at rest.

\subsection{Organization}
The paper is organized as follows.
In Section~\ref{sec.buildingblocks}, we state preliminaries.
In Section~\ref{sec.propyla}, we describe our long-term secure storage architecture \PROP.
In Section~\ref{sec.securityanalysis}, we analyze its security and in Section~\ref{sec.evaluation}, we evaluate its performance.
%Finally, we draw conclusions in Section~\ref{sec.conc}.

%!TEX root = ../Main.tex

\section{Preliminaries}\label{sec.buildingblocks}

%\begin{description}[leftmargin=0cm]\setlength\itemsep{1em}
\subsection{Time-Stamping}
A timestamp scheme \cite{DBLP:conf/crypto/HaberS90} consists of an algorithm $\Setup$, a protocol $\Stamp$ executed between a client $\mathcal{C}$ and a timestamp service $\TS$, and an algorithm $\VerTs$ with the following properties.
\begin{description}
\item[$\Setup()$:]
At initialization, the timestamp service $\TS$ runs this algorithm and outputs a public verification key $pk$.

\item[$\Stamp$:]
On input a data object $\dat$, the client runs this protocol with the timestamp server. When the protocol has finished, the client obtained a timestamp $\tst$ for $\dat$, where $\tst.t$ denotes the time associated with the timestamp.

\item[$\VerTs_{pk}(\dat,\tst)$:]
Given the public key of the timestamp service, on input a data object $\dat$ and a timestamp $\tst$, this algorithm returns $1$, if the timestamp is valid for $\dat$, and $0$ otherwise.
\end{description}
A timestamp scheme is secure if it is infeasible for an adversary to generate a timestamp $\tst$ and a data object $\dat$ such that $\tst$ is valid for $\dat$ and $\dat$ did not exist at time $\tst.t$.
This security notion is formalized in \cite{DBLP:conf/asiacrypt/BuldasS04,DBLP:conf/pkc/BuldasL07,DBLP:conf/acisp/BuldasL13}.

\subsection{Commitments}
A commitment scheme is the digital analog of a sealed envelope and allows to commit to a message $m$ without revealing it.
It consists of the algorithms $\Setup$, $\Commit$, and $\VerCom$, that have the following properties.
\begin{description}
\item[$\Setup()$:]
This algorithm is run by a trusted party and generates the public commitment parameters $cp$.

\item[$\Commit_{cp}(\dat)$:]
Given the commitment parameters $cp$, on input a data object $\dat$, this algorithm outputs a commitment $c$, and a decommitment $d$.

\item[$\VerCom_{cp}(\dat,c,d)$]
Given the commitment parameters $cp$, on input a data object $\dat$, a commitment $c$, and a decommitment $d$, this algorithm outputs $1$ if $d$ is a valid decommitment from $c$ to $\dat$ and $0$ if it is invalid.

\end{description}
A commitment scheme is secure if it is hiding and binding.
Hiding means that a commitment does not reveal anything about the message and binding means that a committer cannot decommit to a different message.
Here, we are interested in information theoretically hiding and computationally binding commitments.
For a more formal description of commitment schemes we refer to \cite{Halevi1996,goldreich2001foundations}.

\subsection{Proactive Secret Sharing}\label{sec.secshare}
A proactive secret sharing scheme \cite{Herzberg1995} allows a client $\mathcal{C}$ to share a secret among a set of shareholders $\SH$ such that less than a threshold number of shareholder cannot learn the secret.
It is defined by protocols $\Setup$, $\Share$, $\Reshare$, and $\Reconstruct$.
%
%Algorithm $\Share$, on input of a message $m$, a number of shareholders $n$, and a threshold number $t$, generates a set of secret shares $s_1,\ldots,s_n$ which are then distributed to the shareholders.
%Here, $t$ is the number of shares required to reconstruct the secret.
%Algorithm $\Reconstruct$, on input of a set of secret shares $s_1,\ldots,s_{n'}$, with $n' > t$, reconstructs the secret message $m$.
%Protocol $\Reshare$ is run between the shareholders and refreshes the shares so that old shares cannot be combined with new shares. The prevents an attacker from gradually breaking into the shareholders and recovering one share after another.
%
\begin{description}
\item[$\Setup$:]
In this protocol the client and the shareholders establish the system parameters.

\item[$\Share$:]
This protocol is run between the client and the shareholders.
The client takes as input a data object $\dat$ and after the protocol has ended each shareholder holds a share of the data.
\item[$\Reshare$:]
This protocol is run periodically between the shareholders. The shareholders take as input their shares and after the protocol is finished, the shareholders have obtained new shares which are not correlated with the old shares.
This protects against a mobile adversary who compromises one shareholder after another over time.
\item[$\Reconstruct$:]
This protocol is run between the shareholders and the client.
The shareholders take as input their shares and after the protocol is finished the client outputs the reconstructed data object $\dat$.
If the reconstruction fails, the client outputs $\bot$.
\end{description}
A secret sharing scheme is secure if no information about the secret is leaked given that at least a threshold number of shareholders is not corrupted. The threshold is usually determined in the setup phase and in some schemes can be changed during the resharing phase.
%An extension of standard secret sharing is proactive secret sharing \cite{Herzberg1995} which additionally allows the shareholders to renew the shares in order to protect against a mobile adversary that compromises one shareholder after another over time.

%\begin{definition}[Era]
%The time elapsed between two sequential commitment updates.
%\end{definition}
%
%\begin{definition}[Epoch]
%The time elapsed between two sequential timestamp updates.
%\end{definition}

\subsection{Oblivious RAM}
An oblivious random access machine (ORAM) \cite{Goldreich:1987:TTS:28395.28416} allows a client to access a remotely stored database such that the storage server does not learn which data items are of current interest to the client.
Here we assume that the client's data consists of $N$ blocks of equal size and each block is associated with a unique identifier $\id \in \{1,\ldots,N\}$.
The server holds a database of size $M>N$ blocks and each block in the server database is identified by a location $i \in \{1,\ldots,M\}$.
A (stash-free) ORAM is defined by algorithms $\Setup$, $\GenAP$, and $\GetId$ with the following properties.
\begin{description}
\item[$\Setup(N)$:]
This algorithm takes as input the client database size $N$ and generates a client local state $s$ and a server database size $M$.

\item[$\GenAP(s,\id)$:]
This algorithm gets as input a client state $s$ and a client block identifier $\id \in \{1,\ldots,N\}$, and outputs an access pattern $\AP \in \{1,\ldots,N\}^{2 \times n}$, which is a sequence of server block location pairs, and a new client local state $s'$.

\item[$\GetId(s,i)$:]
This algorithm gets as input a client state $s$ and a server block location $i \in \{1,\ldots,M\}$, and outputs the corresponding client block identifier $\id \in \{1,\ldots,N\}$.
\end{description}

An ORAM is used as follows to store and access a database of size $N$.
First, the client runs algorithm $\Setup(N) \to (s,M)$ and initializes the server database with $M$ data blocks.
To access (i.e., read or write) block $\id \in \set{1,\ldots,N}$ at the server, the client first computes $\GenAP(s,\id) \to (\AP, s')$ and updates its local state $s \gets s'$.
Then it accesses the server database according to the access pattern $\AP = [(i_1,j_1),\ldots,(i_n,j_n)]$, as follows.
For every block location pair $(i,j) \in \AP$, it first retrieves block $i$.
Then, it checks if $\GetId(s,i)=\id$ and if this is the cae processes the data.
Afterwards, it stores the block at the new location $j$.

An ORAM is secure if the access patterns generated by $\GenAP$ are indindistinguishable from each other.
In the security experiment (Algorithm~\ref{def:aph}), an adversary can instruct the client to access blocks of its choice and then sees the induced access patterns. In order to break the security, the adversary has to distinguish the access patterns of two access instructions of its choice.

\begin{definition}[ORAM Security]\label{def:aph}
An oblivious RAM \oram is information theoretically secure if for any integer $N$ and probabilistic algorithm 
$\calA$,
$$\Pr \left[\Exp{APH}{\oram,N}(\calA)=1 \right] = \frac{1}{2} \text{ .}$$
\begin{algorithm}[h]
\caption{The ORAM access pattern hiding experiment 
$\Exp{APH}{\oram,N}(\calA)$.\label{alg.aph.oram}}
$(s,M) \gets \oram.\Setup(N)$\;
$(\id_1,\id_2) \gets \calA^{\Client}(M)$\;
%\tcc*[h]{$\Access_i = (\op_i,\id_i,\dat_i)$}\;
$b \gets \mathsf{PickRandom}(\{1,2\})$\;
$\AP \gets \Client(\id_b)$\;
$b' \gets \calA^{\Client}(\AP)$\;
\eIf{b=b'}{\KwRet $1$\;}{\KwRet $0$\;}
\vskip 1mm
\fbox{\begin{minipage}{12em}
\underline{\textbf{oracle} $\Client(\id)$:}\\
%$(\Map,\AP) \gets \oram.\Access(\id)$\;
%\KwRet $\AP$;
$(s,\AP) \gets \oram.\GenAP(s,\id)$\;
\KwRet $\AP$;
\end{minipage}}
\vskip 1mm
\end{algorithm}
\end{definition}
%
%We remark that for the purpose of long-term security we are interested in information theoretically secure ORAMs (i.e., secure against computationally unbounded adversaries).

\if{false}
\subsection{\LINCOS: Long-Term Secure Storage}
\label{sec.lincos}

The long-term secure storage system \LINCOS recently proposed by Braun et al.\ \cite{Braun:2017:LSS:3052973.3053043} is the first storage system that combines long-term confidentiality with long-term integrity protection.
Here we briefly summarize its functionality.

\LINCOS comprises the following components:
the client, $\mathcal{C}$, wants to securely store his data.
The confidentiality system, $\ConfSys$, is responsible for storing the confidential data.
It consists of a set of storage servers and uses information theoretically secure secret sharing to achieve long-term confidentiality protection.
The integrity system, $\IntSys$, is responsible for maintaining long-term integrity of the stored data.
%The Confidentiality System is based on secret sharing to protect the confidentiality of the data and consists of a set of Shareholders $\{\SH_1,\SH_2,\ldots\}$.
It consists of an evidence service, $\ES$, and a timestamp service, $\TS$. The evidence service receives commitments to the data from the client and requests timestamps on the commitments from the timestamp service.
Computationally secure authenticated channels secure the communication in the integrity system and information theoretically secure private channels secure the communication in the confidentiality system.
For more details on \LINCOS we refer to the original paper \cite{Braun:2017:LSS:3052973.3053043}.

\fi

% !TEX root = ../Main.tex

\section{Description of \PROP}\label{sec.propyla}

In this section, we describe our new long-term secure storage architecture \PROP, which is the first storage architecture that provides long-term integrity, long-term confidentiality, and long-term access pattern hiding.
%We describe how we modify \LINCOS in order to provide access pattern hiding and to obtain our new long-term secure storage architecture \PROP.

%\todo[inline]{We are giving a generic construction of \PROP. It can use any ORAM that provides information theoretic security and is not restricted to PathORAM.}
%We now describe \PROP, our new architecture for long-term secure storage that supports access pattern hiding.
%\PROP allows to store a set of data blocks $\dat_1,\ldots,\dat_n$.
% and uses an ORAM \oram to hide which blocks are accessed by the client.

\subsection{Overview}
\PROP comprises the following components: a client, an integrity system, which consists of an evidence service and timestamp service, and a confidentiality system, which consists of a set of shareholders (Figure~\ref{fig.overview}).

\begin{figure}
\includegraphics[width=0.9\columnwidth]{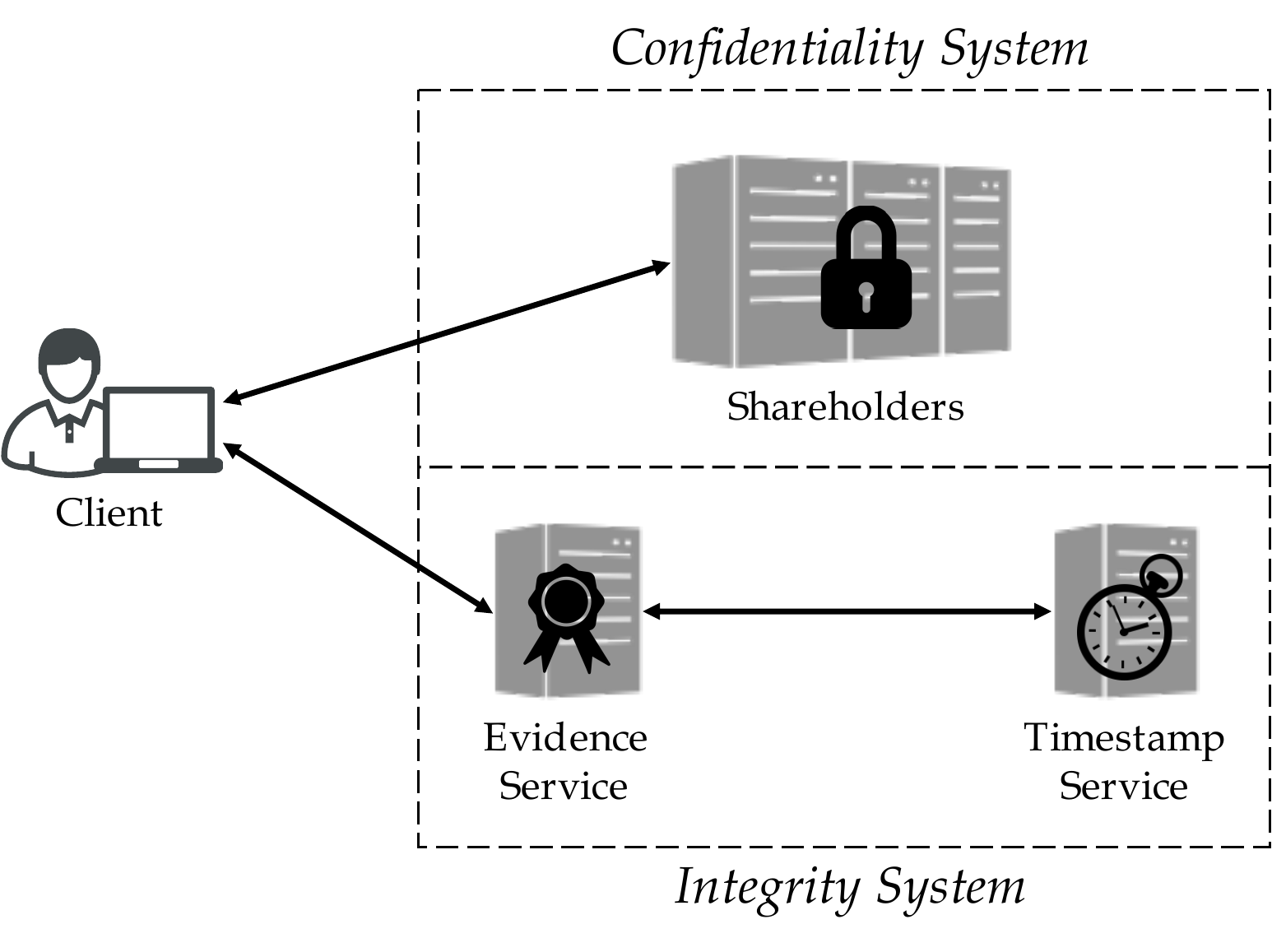}
\caption{Overview of the storage architecture \PROP.\label{fig.overview}\todo[inline]{include ORAM in picture? separate shareholders?}}
\end{figure}

We assume that the client stores a database $D$ that consists of $N$ data blocks that have equal length: $$D=[\dat_1, \dat_2, \ldots, \dat_N] \text{ .}$$
The data is stored at the shareholders using secret sharing, where each block is shared separately.
Integrity protection of the data blocks is achieved by maintaining for each data block $i$ an evidence block $E_i$.
An evidence block $E_i$ has the form
\begin{equation*}
E_i=[(\op_{i,1},c_{i,1},d_{i,1},\tst_{i,1}), (\op_{i,2},c_{i,2},d_{i,2},\tst_{i,2}), \ldots]
\end{equation*}
and describes a history of operations that have been performed on the block.
% (i.e., a read or write operation, a timestamp renewal, or a commitment renewal).
For an element $(\op_{i,j},c_{i,j},d_{i,j},\tst_{i,j})$ of such an evidence block $E_i$, the first element $\op_{i,j} \in \set{\OpWrite,\OpRead,\OpReCom,\OpReTs}$ describes the operation type that has been performed, and the elements $c_{i,j}$, $d_{i,j}$, and $\tst_{i,j}$ refer to the commitment, decommitment, and timestamp, which have been generated during that operation.
Here, the commitments and timestamps form a chain, where later commitments and timestamps guarantee the validity of earlier commitments and timestamps.
The evidence is partially stored at the evidence service and partially stored at the shareholders.
The newest part of the evidence is stored at the evidence service so that the timestamps can be renewed by the evidence service without the help of the data owner.
%A new item is added to $E_i$ whenever the corresponding data block $i$ is touched (i.e., read, written, timestamp renewed, or commitment renewed).

To achieve access pattern hiding, the client makes accesses to the evidence service and the storage servers using an information theoretically secure ORAM.
Therefore, the evidence service and the shareholders store a database consisting of $M>N$ blocks, where $M$ is determined by the choice of the ORAM.
The additional $M-N$ blocks provide the client with the necessary storage space so that it can reshuffle and access the data such that a uniform distribution of accesses over the server blocks is achieved.
In order to preserve the access pattern hiding property, there must also be no correlation between the transmitted data of any two blocks.
Typically, this is achived by re-encrypting the data on every access.
In our solution, however, this technique cannot be applied because the evidence service must receive the commitments in plaintext so that it can timestamp them in order to renew the integrity protection.
Instead, we use a recommitment technique to refresh the commitments:
we let the client commit to the previous commitment and timestamp.
Thereby, a new commitment is obtained that is indistinguishable from other fresh commitments while the connection to the originally timestamped commitment is maintained.

%The evidence service is responsible for generating and storing new evidence.
%When the client accesses a data block, it reads the corresponding evidence block from the evidence service. After an access, the client then stores the existing evidence with the data at the shareholders.
%The evidence service then discards the existing evidence and starts a new evidence list for the given block.
%We denote the evidence currently stored at the evidence service by $E_{\ES,i}$ and the evidence stored at the shareholders by $E_i$.
%
%\PROP consists of the same components as \LINCOS, that is, an Integrity System and a Confidentiality System, where the Integrity System consists of an evidence service and a Timestamp Service and the Confidentiality System consists of a set of shareholders (Figure~\ref{fig.overview}).
%The evidence blocks are partially stored at the shareholders and partially stored and maintained by the evidence service.
%The secret data is stored at the shareholders.

\subsection{Protocols}
%We now describe the functionality of \PROP in more detail.
%\PROP uses an ORAM, a secret sharing scheme, and a series of commitment schemes and timestamp schemes over time.
In the following we describe the protocols used in \PROP for initializing the system, accessing and protecting the data, renewing the protection, and integrity verification.

Throughout the description of the protocols we use the following notation.
For $i \in \set{1,\ldots,M}$ and commitment $c$, we write $\ES.\Write(i,c)$ to denote that the client instructs the evidence service to store commitment $c$ at block $i$.
Furthermore, we write $E_\ES \gets \ES.\Read(i)$ to denote that the client retrieves evidence $E_\ES$ of block $i$ from the evidence service.
Likewise, we denote by $\SH.\Write(i,(\dat,E))$ that the client stores data $\dat$ and evidence $E$ to block $i$ at the shareholders using protocol $\SHARE.\Share$ and we write $(\dat,E) \gets \SH.\Read(i)$ to denote that the client retrieves $\dat$ and $E$ of block $i$ from the shareholders using protocol $\SHARE.\Reconstruct$, where $\SHARE$ is the secret sharing scheme chosen by the client in the initialization phase.
For a data object $\dat$, we write $\tst \gets \TS.\Stamp(\dat)$ to denote that a timestamp $\tst$ for $\dat$ is obtained from timestamp service $\TS$.
Likewise, we write $(c,d) \gets \CSI.\Commit(\dat)$ to denote that a commitment and decommitment are generated for $\dat$ using commitment scheme instance $\CSI$.
Throughout our description, we assume that the timestamp services are initialized appropriately using algorithm $\Setup$ of the corresponding timestamp scheme.
Similarly, we assume that commitment scheme instances are initialized by a trusted third party using algorithm $\Setup$ of the corresponding commitment scheme.
Also, we use the following notation for lists.
We write $[]$ to denote an empty list. For a list $E=[x_1,\ldots,x_n]$ and $i \in \{1,\ldots,n\}$, by $E[i:]$ we denote the sublist $[x_i,\ldots,x_n]$, by $E[:-i]$ we denote the sublist $[x_1,\ldots,x_{n-i+1}]$, and by $E[-i]$ we denote the element $x_{n-i+1}$.

\subsubsection{Initialization}
At initialization, the client chooses a secret sharing scheme $\SHARE$, an ORAM scheme $\oram$, and a database size $N$.
It then initializes the ORAM via $(s,M) \gets \oram.\Setup(N)$ and allocates a database with $M$ blocks at the evidence service.
For each $i \in \set{1,\ldots,M}$, the evidence service initializes block $i$ with an empty evidence lists, $E_{\ES,i}=[]$.
Afterwards, the client picks a set of secret share holders and a reconstruction threshold, and then initializes the secret sharing database with $M$ blocks using protocol $\SHARE.\Setup$.
The shareholders initialize their databases such that for each $i \in \{1,\ldots,M\}$ an empty data object $\dat_i=\bot$ and an empty evidence list $E_i=[]$ is stored at block $i$.
We remark that while we use a stash-free ORAM model for the benefit of a more comprehensible description, the construction can be adopted to work with stashed ORAMs in which case the client locally manages the stashed items as usual.

\subsubsection{Read and Write}
The client reads and writes data blocks using algorithm $\Access$ (see Algorithm~\ref{alg.access}).
This algorithm gets as input an operation type $\op \in \{\OpWrite,\OpRead\}$, a block identifier $\id \in \{1,\ldots,N\}$, optionally data to be written $\dat'$, the ORAM state $s$, a commitment scheme instance $\CSI$, and a reference to a timestamp service $\TS$.
It then generates an access pattern, $(\AP,s) \gets \allowbreak \oram.\GenAP(s,\id)$, and for each $(i_k,j_k) \in \AP$ does the following:
first, it retrieves the new evidence $E_{\ES}$ for block $i_k$ from the evidence service and it retrieves the stored data $\dat$ and the old evidence $E$ from the shareholders.
Then the new evidence $E_{\ES}$ is added to the shareholder evidence $E$.
%Now it is distinguished whether the block should be read or written.
If this is a write operation ($\op = \OpWrite$) and $\oram.\GetId(s,i_k)=\id$, then the block data $\dat$ is replaced with $\dat'$ and a commitment $(c,d)$ to the new data is generated. Since this block is newly written, the existing evidence is discarded and the corresponding evidence is set to $E=[(\OpWrite,c,d,\bot)]$, where $\bot$ is a placeholder for the timestamp that will later be retrieved by the evidence service.
If this is a read operation ($\op = \OpRead$) and $\oram.\GetId(s,i_k)=\id$, then the data $\dat$ is cached and returned when the access algorithm finishes.
Finally, the algorithm refreshes the commitment by creating a new commitment to the block data and the previous commitment.
This is necessary so that the evidence service cannot trace how the client rearranges the blocks.
The refreshed commitment is stored at the evidence service at the new location $j_k$.
The evidence service then timestamps the new commitment and stores the timestamp together with the commitment (Algorithm~\ref{alg.eswrite}).
Also, the client generates new secret shares of the data $\dat$ and the shareholder evidence $E$ and stores them at the shareholders at the new location $j_k$.
As the secret shares are newly generated, they do not correlate with the old shares and their relocation cannot be traced either.

\begin{algorithm}
	\caption{$\Access(\op, \id, \dat', s, \CSI, \TS)$, run by the client.\label{alg.access}}
	
%	\tcc{compute access pattern: $\id_k$ is client block, $i_k$ is server block read, $j_k$ is server block written}
	$([(i_1,j_1),\ldots,(i_n,j_n)],s') \gets \oram.\GenAP(\id)$\;
	
	\vskip 1mm
	
	\For{$k=1,\ldots,n$}{
		\tcp{process $k$-th entry in access pattern}
		
		\vskip 1mm
		
		$E_\ES \gets \ES.\Read(i_k)$;
		$(\dat,E) \gets \SH.\Read(i_k)$\tcp*{read data from location $i_k$}
		
		\vskip 1mm
		
		$E[-1].\tst \gets E_\ES[1].\tst$;
		$E \pluseq E_\ES[2:]$\tcp*{move evidence from ES to SH}
			
		\vskip 1mm
		
		\uIf {$\op = \OpWrite$ {\bf and} $\Loc(i_k) = \id$}{
			$\dat \gets \dat'$;
			$(c,d) \gets \CSI.\Commit(\dat')$;
			$E = [(\OpWrite,c,d,\bot)]$\tcp*{write and commit new data, discard old evidence}
		} \ElseIf {$\op = \OpRead$ {\bf and} $\Loc(i_k) = \id$}{
			$\dat'' \gets \dat$;
			$E'' \gets E$\tcp*{save for output later}
		}

		\If {$\op \neq \OpWrite$ {\bf or} $\Loc(i_k) \neq \id$} {
			
			\tcp{refresh commitments}
%			\eIf {${E[-1].\op = \OpRead}$}{
%				\tcp{replace commitments}
%				$(c,d) \gets \CSI.\Commit([E[-2].c, E[-2].\tst])$\;
%				$E[-1].c \gets c$;
%				$E[-1].d \gets d$\;
%			}{
%				\tcp{add new commitments}
%				$(c,d) \gets \CSI.\Commit([E[-1].c, E[-1].\tst])$\;
%				$E \pluseq [(\OpRead,c,d,\bot)]$\;
%			}
			\If {${E[-1].\op = \OpRead}$}{
				$E \gets E[:-2]$\;
			}
			$(c,d) \gets \CSI.\Commit([E[-1].c, E[-1].\tst])$\;
			$E \pluseq [(\OpRead,c,d,\bot)]$\;
		}
		
		$\ES.\Write(j_k,E[-1].c,\TS)$;
		$\SH.\Write(j_k,(\dat,E))$\tcp*{write data to location $j_k$}
		
	}
	\KwRet $(s',\dat'',E'')$\;
\end{algorithm}

%When the evidence service receives a $\Write$ instruction from the client consisting of a commitment $c$ and a location $i$, it performs algorithm $\ES.\Write(i,c)$ (Algorithm~\ref{alg.eswrite}). It first obtains a timestamp $\tst$ for $c$, and then sets evidence block $E_{\ES,i} = [(\bot,c,\bot,\tst)]$, where the $\bot$ symbols are placeholders for the operation type and the decommitment value which will be filled later by the client.

\begin{algorithm}
	\caption{$\ES.\Write(i,c,\TS)$, run by the evidence service.\label{alg.eswrite}}
	\label{alg:rcommit}
	$\tst \gets \TS.\Stamp(c)$\;
	$E_{\ES,i} \gets [(\bot,c,\bot,\tst)]$\;
\end{algorithm}

\subsubsection{Renew Timestamps}
If the security of the currently used timestamp scheme is threatened, the evidence service renews the evidence as follows (Algorithm~\ref{alg.rets}).
It picks a new timestamp service $\TS$ that uses a more secure timestamp scheme and then for every $i \in [1,\ldots,M]$,
it first creates a commitment $(c',d')$ to the last commitment and timestamp stored in $E_{\ES,i}$, then requests a timestamp $\tst$ for $c'$, and finally adds $(\OpReTs,c',d',\tst)$ to $E_{\ES,i}$.

\begin{algorithm}
	\caption{$\RenewTs(\CSI,\TS)$, run by the evidence service.\label{alg.rets}}
	\label{alg:renewtimestamps}
	\For {$i=1$ \TO $M$}{
		$(c_i,d_i) \gets \CSI.\Commit([E_{\ES,i}[-1].c , E_{\ES,i}[-1].\tst])$\;
		$ts_i \gets \TS.\Stamp(c_i)$\;
		$E_{\ES,i} \pluseq [(\OpReTs,c_i,d_i,\tst_i)]$\;
	}
\end{algorithm}

\subsubsection{Renew Commitments}
If the security of the currently used commitment scheme is threatened, the client renews the evidence using Algorithm~\ref{alg.recom}.
It starts by selecting a new commitment scheme instance $\CSI$.
Then, for each block $i \in [1,\ldots,M]$, it does the following.
It first retrieves the new evidence $E_{\ES,i}$ from the evidence service and data block $\dat_i$ and old evidence block $E_i$ from the shareholders.
It then adds the new evidence $E_{\ES,i}$ to the shareholder evidence $E_i$.
Afterwards, it creates a new commitment $(c_i,d_i)$ to the secret data $\dat_i$ and the evidence $E_i$.
It then adds $(\OpReCom,c_i,d_i,\bot)$ to $E_i$.
Finally, it sends $c_i$ to the evidence service and distributes $\dat_i$ and $E_i$ to the shareholders.

\begin{algorithm}
	\caption{$\RenewCom(\CSI,\TS)$, run by the client.\label{alg.recom}}
	
	\For{$i=1,\ldots,M$}{
		$E_\ES \gets \ES.\Read(i)$;
		$(\dat,E) \gets \SH.\Read(i)$\tcp*{retrieve data from location $i$}
		
		\vskip 1mm
		
		$E[-1].\tst \gets E_\ES[1].\tst$;
		$E \pluseq E_\ES[2:]$\tcp*{move evidence from ES to SH}
		
		$(c,d) \gets \CSI.\Commit([\dat , E])$\tcp*{recommit to data and evidence}
		$E \pluseq [(\OpReCom,c,d,\bot)]$\tcp*{add new evidence}
		
		\vskip 1mm
		
		$\ES.\Write(i,E[-1].c,\TS)$\;
		$\SH.\Write(i,(\dat,E))$
		\tcp*{store evidence and data at location $i$}
	}
\end{algorithm}

\subsubsection{Share Renewal}
In regular time intervals the shareholders renew the stored shares to protect against a mobile adversary (see Section~\ref{sec.secshare}) by running protocol $\SHARE.\Reshare$.

\subsubsection{Verification}
The verification algorithm (Algorithm~\ref{alg.ver}) uses a trust anchor $\TA$ that certifies the validity of public keys for timestamps and commitments.
Here, by $\VerTs_\TA(\dat,\tst;t_\ver)=1$ we denote that timestamp $\tst$ is valid for $\dat$ at reference time $t_\ver$, and by $\VerCom_\TA(\dat,c,d;t_\ver)=1$ we denote that $d$ is a valid decommitment from $c$ to $\dat$ at time $t_\ver$.
On input a data object $\dat$, a time $t$, an evidence block $E$, and the verification time $t_\ver$, the verification algorithm of \PROP checks whether $E$ is currently valid evidence for the existence of $\dat$ at time $t$, given that the current time is $t_\ver$.
In particular, the algorithm checks that the evidence is constructed correctly for $\dat$, that the timestamps and commitments have been valid at their renewal time, and that the first timestamp refers to time $t$.
%In Section~\ref{sec.ltint.security} we show that this verification procedure satisfies the requirements of long-term integrity protection.

\begin{algorithm}
	\caption{$\VerInt_{\TA}(\dat,t,E;t_\ver)$, run by any verifier.\label{alg.ver}}
	\tcc{verifies that $\dat$ existed at time $t$}
	Let $E=[(\op_1,c_1,d_1,\tst_1),\ldots,(\op_n,c_n,d_n,\tst_n)]$\;
%	\BlankLine
	Set $t_{n+1} := t_\ver$\;
	For $i \in [n]$, set $E_i := [(\op_1,c_1,d_1,\tst_1),\ldots,(\op_i,c_i,d_i,\tst_i)]$\;
	For $i \in [n]$, set $t_i := \tst_i.t$\;
	For $i \in [n]$, set $t_{\mathsf{NRC}(i)} := \min (\{t_j \mid act_j=\OpReCom \land j>i\} \cup \{t_{n+1}\})$\;
	\BlankLine
	\For {$i=n$ \TO $1$}{
		Assert $\VerTs_{\TA}(c_i,\tst_i;t_{i+1}) = 1$\;
		\uIf {$\op_i=\OpWrite$ {\bf and} $i=1$}{
			Assert $\VerCom_{TA}(\dat,c_1,d_1; t_{\mathsf{NRC}(1)}) = 1$\;
		}
		\uElseIf {$\op_i=\OpRead$ {\bf or} $\op_i=\OpReTs$}{
			Assert $\VerCom_{\TA}([c_{i-1} , \tst_{i-1}],c_i,d_i; t_{\mathsf{NRC}(i)}) = 1$\;
		}
		\uElseIf {$\op_i=\OpReCom$}{
			Assert $\VerCom_{\TA}([\dat , E_{i-1}],c_i,d_i; t_{\mathsf{NRC}(i)}) = 1$\;
		}
		\Else{Fail\;}
	}
	Assert $\tst_1.t = t$\;
\label{alg:verint}
\end{algorithm}

%!TEX root = ../main.tex

\section{Security Analysis}\label{sec.securityanalysis}

%\subsection{Terminology issues}
%
%Important questions:\\
%What does Access-Pattern-Hiding mean?\\
%Do we stick to the original connotation by Goldreich and Ostrovsky, that is, APH only concerns the memory locations?
%\nk{it would be nice to simply stick to this, but we have to argue that 
%everything else `seen' during the protocol, does not leak any information 
%regarding the data accessed.}
%If yes, how do we call the notion of security that is achieved if we additionally hide the values (data) and the instructions (read, write)?

In this section we analyze the security of \PROP.
Our security analysis requires a model of real time which is used for expressing the scheduling of protection renewal events in our security experiments and for expressing computational bounds on the adversary with respect to real time.
We first describe our model of real time.
Then, we show that \PROP provides long-term access pattern hiding, long-term confidentiality, and long-term integrity.

\subsection{Model of Time}
\label{sec.timemodel}
%A security model for a cryptographic scheme formally captures the setting in which the scheme is used.
For modeling the security of \PROP, we want to be able to express that certain events (e.g., renewal of timestamps) are performed according to a timed schedule.
Concretely, in the security analysis of \PROP, we consider a \emph{renewal schedule} $\calS$ that describes at which times, and using which schemes the timestamp and commitment renewals are performed.
Additionally, our computational model allows to capture that an adversary becomes computationally more powerful over time (e.g., it gets access to a quantum computer).
%To express this, we include a notion of real time in our security experiments as in \cite{canetti2008modeling,geihs2016security}.
%For example, this is necessary to model that the computational power of an adversary increases over time \cite{geihs16security} or to express that certain events in a security game are performed according to a time schedule.

We model real time as proposed in \cite{geihs2016security}, i.e., we use a global clock that advances whenever the adversary performs work.
Formally, in a security experiment with an adversary $\calA$, we assume a global clock $\Clock$ with a local state $\thetime$ determining the current time in the experiment.
The adversary $\calA$ is given the ability to advance time by calling $\Clock(t)$, i.e., at $\thetime=t$, it may call $\Clock(t')$ and set the time to $\thetime=t'$, for any $t'$ with $t'>t$.
When this happens, all actions scheduled between times $t$ and $t'$ are performed in order and afterwards the control is given back to the adversary.
We remark that the adversary also uses up his computational power when it advances time.
This is due to the fact that we restrict it to perform only a bounded number of operations per time interval.
Our computational model also captures adversaries who increase their computational power over time. For more details on the adversary model we refer the reader to \cite{Buldas2017}.

\subsection{Long-Term Access Pattern Hiding}

In the following we prove that \PROP achieves information theoretically secure access pattern hiding against the evidence service and the shareholders if the used ORAM, secret sharing scheme, and commitment schemes are information theoretically secure.

Formally, Access-Pattern-Hiding (APH) Security of \PROP is defined via game $\Exp{APH}{\PROP}$ (Algorithm~\ref{alg.aphes}), where an adversary, 
$\calA$, instructs the client of \PROP to read and write database blocks at logical addresses chosen by the adversary.
During these data accesses, $\calA$ observes the data that is transferred between the client and the evidence service and a subset of less than the threshold number of shareholders.
At some point in time, $\calA$ gives two different access instructions to the client.
The client picks one of them at random and executes it.
The goal of the adversary $\calA$ is to infer from the observed data stream which of the access instructions has been executed by the client.

For \PROP to be information theoretically secure access pattern hiding, it must hold that for any timestamp and commitment renewal schedule $\calS$, a computationally unbounded adversary is not able to infer with probability other than $\frac{1}{2}$ which access instruction was made.

%At the first stage of that game, $\calA$ may instruct the client to make accesses to the storage system by calling oracle $\Client(\op,\id,\dat)$.
%Then, $\calA$ outputs two queries $(\op_1,\id_1,\dat_1)$ and $(\op_2,\id_2,\dat_2)$ in the challenge phase.
%One of the queries is selected at random and executed by the client.
%At the second stage, $\calA$ may again instruct the client to make accesses of its choice to the storage system.
%At some point, $\calA$ guesses which of its queries was executed in the challenge 
%phase. The adversary wins if it guesses correctly and loses otherwise.
%For \PROP to provide information theoretically secure access pattern hiding, the adversary must not be able to win this game with probability other than $\frac{1}{2}$ for any timestamp and commitment renewal schedule $\calS$.

\begin{definition}[APH-Security of \PROP]
\PROP is information theoretically APH-secure if for any renewal schedule $\calS$ and any adversary $\calA$:
\begin{equation*}
\Pr\left[\Exp{APH}{\PROP}(\calS,\calA)=1\right] = \frac{1}{2} \text{ .}
\end{equation*}
\begin{algorithm}[ht]
\caption{$\Exp{APH}{\PROP}(\calS,\calA)$}
\label{alg.aphes}
%the evidence service 
%
%\vskip 1mm
%
%\begin{minipage}{8.0cm}

%Perform $\PROP.\Init$ and let $\VS$ denote the sequence of views yielded by 
%running multiple instances of the $\Access$ protocol\;
%\todo[inline]{What do you mean by multiple instances?}
%$((\op_1,\id_1,\dat_1), (\op_2,\id_2,\dat_2)) \gets \calA^{\Clock, \Client}(\VS)$\;
%%\tcc*[h]{$\Access_i = (\op_i,\id_i,\dat_i)$}\;
%$b \pickrandom \{1,2\}$\;
%$\V_b \gets \Client(\op_b,\id_b,\dat_b)$\;
%$b' \gets \calA^{\Clock, \Client}(\V_b)$\;
%\eIf{$b=b'$}{\KwRet $1$\;}{\KwRet $0$\;}

%Run $\PROP.\Init$ and let $\AP$ denote the access pattern during 
%initialization\;
$((\op_1,\id_1,\dat_1),(\op_2,\id_2,\dat_2))\gets\calA^{\Clock,\Client}()$\;
$b \gets \mathsf{PickRandom}(\{1,2\})$\;
$\V \gets \Client(\op_b,\id_b,\dat_b)$\;
$b' \gets \calA^{\Clock,\Client}(\V)$\;
\eIf{$b=b'$}{\KwRet $1$\;}{\KwRet $0$\;}
\vskip 1mm
\fbox{\begin{minipage}{24em}
\underline{\textbf{oracle} $\Clock(t)$:}\\
\If{$t>\thetime$}{
\emph{Perform all renewals scheduled in $\calS$ between $\thetime$ and $t$}\;
$\thetime \gets t$\;}
%Let $\AP$ denote the access pattern during these renewals\;
%\KwRet $\AP$;
\end{minipage}}
\vskip 1mm
\fbox{\begin{minipage}{24em}
\underline{\textbf{oracle} $\Client(\op,\id,\dat)$:}\\
%\If{$\op = \OpRead$}{
%$\V \gets \PROP.\Access(\OpRead,\id,\bot)$;
%}
%\If{$\op = \OpWrite$}{
%$\V \gets\PROP.\Access(\OpWrite,\id,\dat)$;
%}
%Let $\V$ denote the access pattern during this execution of \Access;\\
$\PROP.\Access(\op,\id,\dat)$\;
Let $\V$ denote the data received by the evidence service and a subset of less than the threshold number of shareholders during the execution of $\PROP.\Access$;\\
\KwRet $\V$\;
\end{minipage}}
\vskip 1mm
\end{algorithm}
\end{definition}
%
%The following theorem states that \PROP is information theoretically APH-secure if the used ORAM, the used commitment schemes, and the used secret sharing are information theoretically secure.
%
\begin{theorem}\label{thm:itap}
If \PROP is instantiated using an information theoretically secure ORAM, information theoretically hiding commitment schemes, and information theoretically secure secret sharing, then it provides information theoretic APH-security.
%(The proof can be found in Appendix~\ref{app.proofs}.)
\end{theorem}

\begin{proof}
We observe that the queries made by the client and observed by the adversary are either of the form $(i,c)$, when the client instructs the evidence service to store commitment $c$ at database location $i$, or of the form $(i,s)$, when the client instructs a shareholder to store share $s$ at location $i$.
Furthermore, we observe that when using an information theoretically secure ORAM, there is no statistical correlation between the instructions $(\op,\id,\dat)$ chosen by the adversary and the database locations $i$ sent by the client.
There is also no statistical correlation between the data and the commitments $c$ or the secret shares $s$, as long as the adversary observes less than the threshold number of shareholders.
This is (1) because we use information theoretically hiding commitments and information theoretically secure secret sharing and (2) the shares and the commitments are renewed on every access.

As there is no statistical correlation between the accesses made by the client and the transmitted data, even a computationally unbounded adversary cannot infer which of the challenge instructions $(\op_1,\id_1,\dat_1)$ and $(\op_2,\id_2,\dat_2)$ was executed.
\end{proof}

\subsection{Long-Term Confidentiality}\label{sec.ltconf}
Informally, long-term confidentiality of \PROP means that even an unbounded evidence service and a subset of colluding unbounded shareholders cannot learn anything about the content of the stored data.
More formally, we require that there is no significant statistical correlation between the data stored by the client and the data observed by the evidence service and a subset of shareholders.
We observe that this property immediately follows from the information-theoretic access pattern hiding security of \PROP.

%Next, we show that \PROP provides information theoretic confidentiality if the used secret sharing scheme and the commitment schemes are information theoretically hiding.
%Concretely, we show that neither the evidence service, nor the shareholders learn anything about the content of the stored data.
%
%\begin{theorem}\label{thm.ltconf}
%If the used secret sharing scheme and commitment schemes are information theoretically hiding, then \PROP achieves information theoretically secure confidentiality protection, that is, neither the evidence service, nor the shareholders learn anything about the content of the stored data.
%%(The proof can be found in Appendix~\ref{app.proofs}.)
%\end{theorem}
%
%\begin{proof}
%The evidence service does not learn anything about the data content because it only receives information theoretically hiding commitments.
%Furthermore, from the information theoretic security of the secret sharing scheme it follows that also the shareholders do not learn anything about the shared data.
%%To protect against an adversary who taps the network channels we use information theoretically secure private channels between the client and the shareholders.
%\end{proof}

\subsection{Long-Term Integrity}\label{sec.ltint.security}
Next, we show that \PROP provides long-term integrity protection.
Here we consider an adversary that may be running for a very long time but who can only perform a limited amount of work per unit of time (see the adversary model description in Section~\ref{sec.timemodel}).

The goal of the adversary is to prove that a data object existed at a certain point in time when in reality it did not exist.
To capture this notion more formally, we need to define what it means that a data object existed.
%We say \PROP achieves long-term integrity protection if any adversary $\calA$ fails with high probability to produce valid evidence $E$ for a data object $\dat$ and time $t$, if the chosen commitment and timestamp schemes are secure against $\calA$ within their usage period and $\dat$ did not exist at $t$.
Here, when we say that a data object $\dat$ existed at time $t$ with respect to an adversary $\calA$, we mean that $\calA$ ``knew'' $\dat$ at time $t$ (with high probability).
This can be formalized using computational extractors (e.g., \cite{DBLP:conf/pkc/BuldasL07,geihs2016security,Buldas2017}).
Our security analysis is based on \cite{geihs2016security,Buldas2017}, where it is shown that (under certain computational assumptions) extractable commitments and timestamps can be used to argue about the knowledge of an adversary at a given point in time.
Based on these results, we here give a slightly less technical security proof for the integrity of \PROP.

In our security analysis we use the following notation to describe the knowledge of an adversary $\calA$.
For a data object $\dat$ and a time $t$, we write $\dat \in \calK_\calA[t]$ to denote that $\calA$ knew $\dat$ at time $t$.
We assume without loss of generality that adversaries do not forget knowledge, that is, for any data object $\dat$ and any two points in time $t$ and $t'$, if $\dat \in \calK_\calA[t]$ and $t' > t$, then $\dat \in \calK_\calA[t']$.
We also use the convention that for verification of timestamps and commitments a trust anchor $\TA$ is provided by a PKI that certifies the verification keys of the used timestamp and commitment scheme instances and specifies the corresponding instance validity periods.

We state two lemmas that will be useful for proving long-term integrity protection of \PROP. These lemmas are derived from results of \cite{geihs2016security,Buldas2017} about extractable timestamps and commitments.
The first lemma states that if an adversary $\calA$ knows a timestamp $\tst$ and a message $m$ at a time $t$, and $\tst$ is valid for $m$ at $t$, then $\calA$ has already known $m$ at time $\tst.t$ (with high probability).
The second lemma states that if an adversary $\calA$ knows a commitment value at a time $t$, a message $m$ and a decommitment $d$ are known at time $t'>t$, and $d$ is a valid decommitment at time $t'$, then $\calA$  has already known the message $m$ at the commitment time $t$ (with high probability).\todo{rename lemma and theorem to argument, proof to proof sketch}
\begin{lemma}\label{lem.ts}
For any message $m$, timestamp $\tst$, and time $t$:
\begin{equation*}
(m,\tst) \in \calK_\calA[t] \land \VerTs_\TA(m,\tst;t)=1 \implies m \in \calK_\calA[\tst.t] \text{ .}
\end{equation*}
\end{lemma}
\begin{lemma}\label{lem.com}
For any commitment value $c$, time $t$, message $m$, decommitment value $d$, and time $t'>t$:
\begin{multline*}
c \in \calK_\calA[t] \land (m,d) \in \calK_\calA[t'] \land \VerCom_\TA(m,c,d;t')=1 \\ \implies m \in \calK_\calA[t] \text{ .}
\end{multline*}
\end{lemma}
%
%The validity of these lemmas builds on the security notions of extractable time-stamping and extractable commitments.
%A more formal description of these lemmas and corresponding proofs are discussed in \cite{buldas2007knowledge,geihs2016security}.
%
Next, we use Lemma~\ref{lem.ts} and Lemma~\ref{lem.com} to show that \PROP provides long-term integrity protection.
\begin{theorem}\label{thm.ltint}
\PROP provides long-term integrity protection, that is, it is infeasible for an adversary to produce evidence $E$ valid for data $\dat$ and time $t$ without having known $\dat$ at time $t$ given that the used timestamp and commitment schemes are secure within their usage period.
%(The proof can be found in Appendix~\ref{app.proofs}.)
\end{theorem}

\begin{proof}
Assume an adversary $\calA$ outputs $(\dat,t,E)$ at some point in time $t_{n+1}$ and that $\TA$ is the trust anchor provided by the PKI at that time.
We show that if $E$ is valid evidence for data $\dat$ and time $t$ (i.e., $\VerInt_\TA(\dat,t,E)=1$), then the adversary did know $\dat$ at time $t$ (with high probability).

Let $E=[(\op_1,c_1,d_1,\tst_1),\ldots,(\op_n,c_n,d_n,\tst_n)]$ without loss of generality.
Then, for $i \in [1,\ldots,n]$, define $E_i=[(\op_1,c_1,d_1,\tst_1),\allowbreak\ldots,\allowbreak(\op_i,c_i,d_i,\tst_i)]$, $t_i = \tst_i.t$, and $\tnrc{i}$ as the time of the next commitment renewal after commitment $c_i$.
Additionally, we define $\tnrc{n} = t_{n+1}$.
In the following, we show recursively that for $i \in [n,\ldots,1]$, statement
\begin{equation*}
St(i): (c_i,\tst_i) \in \calK_\calA[t_{i+1}] \land (\dat,E_i) \in \calK_\calA[\tnrc{i}]
\end{equation*}
holds, that is, commitment value $c_i$ and timestamp $\tst_i$ are known at the next timestamp time $t_{i+1}$ and the data $\dat$ and partial evidence $E_i$ are known at the next commitment renewal time $\tnrc{i}$.
%Finally we observe that statement $St(2)$ implies $\dat \in \calK_\calA[t_1]$, i.e., $\dat$ existed at time $t_1$.

Statement $St(n)$ obviously holds because $(\dat,E)$ is output by the adversary at time $t_{n+1}$ and includes $c_n$, $\tst_n$, $\dat$, and $E_n$.
Next, we show for $i \in \{n,\ldots,2\}$, that $\VerInt_\TA(\dat,t,E)=1$ and $St(i)$ implies $St(i-1)$.
%
%Now assume statement $St(i)$ holds and $\VerInt(\dat,t,D)=1$.
By the definition of $\VerInt$ (Algorithm~\ref{alg.ver}) we observe that $\VerInt_\TA(\dat,t,E)=1$ implies $\VerTs_\TA(c_i,\tst_i;t_{i+1}) = 1$, that is, $\tst_i$ is valid for commitment $c_i$ at time $t_{i+1}$. Furthermore, we observe that $St(i)$ implies $(c_i,\tst_i) \in \calK_\calA[t_{i+1}]$, which means that $c_i$ and $\tst_i$ were known at time $t_{i+1}$.
We can now apply Lemma~\ref{lem.ts} to obtain that commitment $c_i$ was known at time $t_i$ (i.e., $c_i \in \calK_\calA[t_i]$).
%
%Furthermore, if $\op_i \in \{\OpRead,\OpReTs\}$ and the case where $\op_i=\OpReCom$.
Next, we distinguish between the case $\op_i \in \{\OpRead,\OpReTs\}$ and the case $\op_i \in \{\OpReCom\}$.
We observe that if $\op_i \in \{\OpRead,\OpReTs\}$, then $\VerCom_\TA(c_{i-1} \| \tst_{i-1},c_i,d_i; \allowbreak \tnrc{i}) = 1$ and by Lemma~\ref{lem.com} it follows that $(c_{i-1} , \tst_{i-1}) \in \calK_\calA[t_i]$ (i.e., $c_{i-1}$ and $\tst_{i-1}$ were known at time $t_i$).
%We observe that for $\op_i \in \{\OpRead,\OpReTs\}$, we have that $\tnrc{i-1}=\tnrc{i}$ and obtain that $A(i-1): c_{i-1},\tst_{i-1} \in \calK_\calA[t_{i}] \land \dat,D_{i-1} \in \calK_\calA[t_{\text{NextReCom}(i-1)}]$ holds.
%
%Now assume that $\op_i = \OpReCom$.
%We observe that $\VerInt(\dat,t,D)=1$ implies $\VerCom(\dat \| D_{i-1},c_i,d_i; t_{\text{NextReCom}(i)}) = 1$ and $A(i)$ implies $\dat,D_{i-1},d_i \in \calK_\calA[\tnrc{i}]$.
If $\op_i \in \{\OpReCom\}$, then $\VerCom_\TA(\dat \| E_{i-1},c_i,d_i; \allowbreak \tnrc{i}) = 1$ and it follows that $(\dat,E_{i-1})\in \calK_\calA[\tnrc{i-1}]$.
%We observe that for $\op_i = \OpReCom$, we have that $\tnrc{i-1}=t_i$ and obtain that $A(i-1): c_{i-1},\tst_{i-1} \in \calK_\calA[t_{i}] \land \dat,D_{i-1} \in \calK_\calA[t_{\text{NextReCom}(i-1)}]$.
Together, we obtain that if $\VerInt_\TA(\dat,t,E)=1$ and $St(i)$, then $St(i-1)$ holds.
By induction over $i$ it follows that $St(1)$ holds.

Finally, we observe that $St(1)$ by Lemma~\ref{lem.ts} and Lemma~\ref{lem.com} implies $\dat \in \calK_\calA[t_1]$.
Also, $\VerInt_\TA(\dat,t,E)=1$ implies $t_1=t$, and thus we obtain $\dat \in \calK_\calA[t]$.
We conclude that if the adversary presents $(\dat,t,E)$ such that $\VerInt_\TA(\dat,\allowbreak t,E)=1$, then it must have known $\dat$ at time $t$.
\end{proof}

%!TEX root = ../main.tex

\section{Instantiation and Evaluation}\label{sec.evaluation}

%\subsection{Performance}
%\label{sec.perf}

%In this section, we evaluate the performance of \PROP instantiated using Path-ORAM \cite{pathoram}, Pedersen Commitments \cite{Pedersen1992}, Shamir Secret Sharing \cite{Shamir:1979:SS:359168.359176}, and RSA Signatures \cite{Rivest:1978:MOD:359340.359342}.

\newcommand{\postquantum}{1}
%\selectcolormodel{gray}

In this section, we describe an instantiation of \PROP and analyze its performance.

\subsection{Scheme Instantiation}
\label{sec.instantiation}

We instantiate \PROP using Path-ORAM \cite{pathoram}, Shamir Secret Sharing \cite{Shamir:1979:SS:359168.359176}, Halevi-Micali Commitments \cite{Halevi1996}, RSA Signatures \cite{Rivest:1978:MOD:359340.359342}, and XMSS Signatures \cite{Buchmann2011}.
The implementation has been done in Java and we use the following parameters.
For Path-ORAM we use a bucket size of $5$.
Our implementation of Halevi-Micali Commitments uses the hash functions SHA-224, SHA-256, and SHA-384 \cite{sha2}.
For RSA, we use the standard JDK implementation and use SHA-224 with RSA-2048.
For XMSS, we use the implementation from the Bouncy Castle Library \cite{bouncycastle}, which supports the hash functions SHA-256 and SHA-512, and we use a tree height of 10.

This instantiation has the required security properties.
Path ORAM instantiated with an information theoretically secure random number generator (e.g., a quantum-based random number generator \cite{doi:10.1080/09500340008233380}) provides information theoretic security \cite{pathoramfull}.
Shamir Secret Sharing and Halevi-Micali Commitments are information theoretic hiding.
Therefore, by Theorem~\ref{thm:itap}, the described instantiation of \PROP is information theoretic access pattern hiding.
Information theoretic confidentiality also follows from Theorem~\ref{thm:itap} (see Section~\ref{sec.ltconf}).
Finally, by Theorem~\ref{thm.ltint}, long-term integrity is achieved as long as the used commitment and signature scheme instances are secure within their usage period.
In Table~\ref{tab.parameters} we list the commitment and signature scheme instances that we use together with their usage periods, which are based on the predictions by Lenstra and Verheul \cite{DBLP:journals/joc/LenstraV01,lenstra2004key,keylength.com}.
Also, we assume that quantum computers will become a considerable threat by 2040 and therefore transition from RSA Signatures (which are known vulnerable to quantum computers) to XMSS Signatures (which are expected secure against quantum computer attacks) after 2030.

\subsection{Evaluation}

\subsubsection{Scenario}
In the following we examine a use case where a client stores $N$ data objects of size $L$ for a time period of \numprint{100} years using \PROP with 3 shareholders.
To maintain long-term integrity protection, timestamps are renewed every 2 years, which corresponds to the typical lifetime of a public key certificate, and commitments are renewed every 10 years, which corresponds to the longer lifetime of commitments because they do not involve any secret parameters that could compromise security.
Cryptographic schemes in \PROP are instantiated as described in Section~\ref{sec.instantiation}.
Computation times of \PROP are estimated by counting the number of the involved primitive operations (e.g., commitments, timestamps) and multiplying these numbers by the running time of the respective operations, which were measured on a computer with a \SI{2.9}{\giga\hertz} Intel Core i5 CPU.
We remark that our performance analysis does not consider network latency.

\begin{table}
\center

\caption{\label{tab.parameters} Overview of the used commitment and signature scheme instances and their usage period.}

\begin{tabular}{ccc}
\toprule
\textbf{Years} & \textbf{Signatures} & \textbf{Commitments} \\
\midrule

2018-2030 & RSA-2048 & HM-224 \\

2031-2066 & XMSS-256 & HM-224 \\

2067-2091 & XMSS-256 & HM-256 \\

2091-2118 & XMSS-512 & HM-384 \\

\bottomrule
\end{tabular}

%\vspace{2em}
%
%\begin{tabular}{ccc}
%\toprule
%\textbf{Scheme instance} & \textbf{Security (until year)}\\
%\midrule
%
%HM-224 & 2066 \\
%HM-256 & 2090 \\
%HM-384 & 2186 \\
%
%\midrule
%
%XMSS-256 & 2090 \\
%XMSS-512 & 2282 \\
%
%\bottomrule
%\end{tabular}

\end{table}

\subsubsection{Results}
We now present the results of our performance evaluation focusing on the costs of enabling long-term integrity protection. In particular, we measure the costs for generating, communicating, and storing timestamps and commitments.

In Figure~\ref{fig:dbsizes}, we show the space required for storing commitments, decommitments, and timestamps.
We observe that the storage space required per shareholder increases with each commitment renewal and that the magnitude of the increase depends on the commitment and signature scheme parameters.
With respect to the evidence service we observe that the required storage space increases with each timestamp renewal and is reset with each commitment renewal, when commitments and timestamps are moved to the shareholders.
The storage space required for integrity protection is independent of the data block size because commitment and signature sizes are independent of the data size.
After 100 years, about \SI{50}{\kilo\byte} of storage per block is required at the evidence service and about \SI{150}{\kilo\byte} at a shareholder.

\begin{figure}
\centering

\begin{tikzpicture}
\begin{axis}[
    xlabel={Year},
    ylabel={\SI{}{\kilo\byte}},
    ymajorgrids=true,
    grid style=dashed,
    width=\columnwidth,
    height=4.5cm,
    x tick label style={/pgf/number format/.cd, set thousands separator={}},
    legend pos=north west,
%    legend columns=-1,
%    legend style={at={(0.5,-0.3)},anchor=north,legend columns=-1},
%    cycle list name=black white,
    y filter/.code={\pgfmathparse{#1/1024}\pgfmathresult},
]

\addplot+
[jump mark left, each nth point=2, mark size=1pt]
%[no marks]
table[x index=0, y index=1]
{data/storage_space.dat};
\addlegendentry{Evidence Service}

\addplot+
[jump mark left, each nth point=10, mark size=1pt]
%[no marks]
table[x index=0, y index=2]
{data/storage_space.dat};
\addlegendentry{Shareholder}

\end{axis}
\end{tikzpicture}

\caption{Storage costs for timestamps and commitments per server block (independent of block size $L$).\label{fig:dbsizes}}
\end{figure}
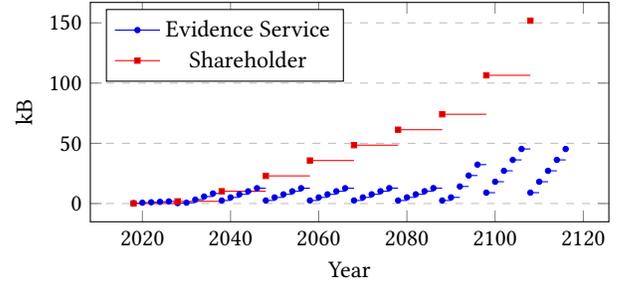

In Figure~\ref{fig:renew}, we show the computation time required for renewing timestamps and commitments for database size $N\in\{2^8,2^{12},2^{16}\}$ and data block size $L=100\si{\kilo\byte}$.
We observe that the renewal time scales linearly with the number of data objects and depends on the commitment and signature scheme parameters.
Renewing the protection after 100 years takes about \SI{1}{\minute} for $N=2^8$ data objects and \SI{256}{\minute} for $N=2^{16}$ data objects.

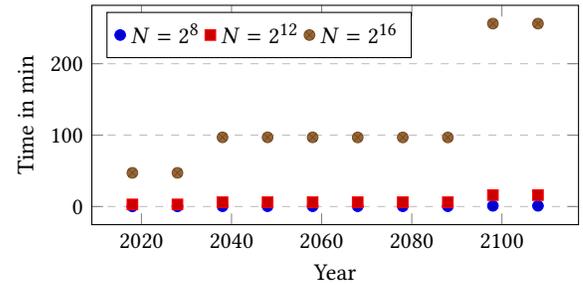
\begin{figure}
\centering

\begin{tikzpicture}
\begin{axis}[
    xlabel={Year},
    ylabel={Time in \si{\minute}},
    ymajorgrids=true,
    grid style=dashed,
    width=0.95\columnwidth,
    height=4.5cm,
    legend columns=-1,
    x tick label style={/pgf/number format/1000 sep=},
    legend style={legend pos=north west},
%    ybar,
%    bar width=2pt,
%    cycle list name=black white,
    y filter/.code={\pgfmathparse{#1/1000/60}\pgfmathresult},
    only marks,
    each nth point=10,
]

\addplot
table[x index=0, y index=1]
{data/renew_time.dat};
\addlegendentry{$N=2^8$}

\addplot
table[x index=0, y index=2]
{data/renew_time.dat};
\addlegendentry{$N=2^{12}$}

\addplot
table[x index=0, y index=3]
{data/renew_time.dat};
\addlegendentry{$N=2^{16}$}

\end{axis}
\end{tikzpicture}

\caption{Computation time for renewing timestamps and commitments when storing $N\in\{2^8,2^{12},2^{16}\}$ data objects of size $L=100\si{\kilo\byte}$.\label{fig:renew}}
\end{figure}

In Figure~\ref{fig:accesstime}, we show the computation time required for generating commitments and timestamps during one database access made by the client for database size $N\in\{2^8,2^{12},2^{16}\}$ and block size $L=100\si{\kilo\byte}$.
We observe that the computation cost scales logarithmically with the database size $N$, which is because Path-ORAM requires $\log(N)$ server accesses per client access.
Furthermore, we observe that the computation cost also significantly depends on the commitment and timestamp scheme parameters.
Concretely, the computation cost per access ranges from about \SI{0.3}{\second} in 2018 for $N=2^8$ and about \SI{3.7}{\second} in 2118 for $N=2^{16}$.

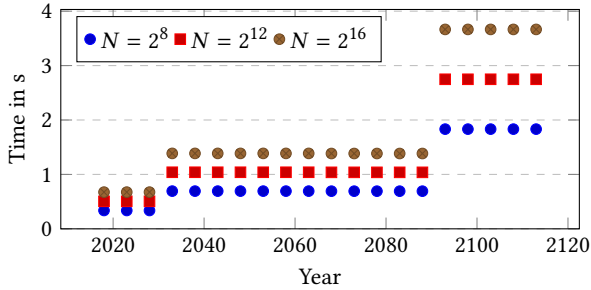
\begin{figure}
\centering

\begin{tikzpicture}
\begin{axis}[
    xlabel={Year},
    ylabel={Time in \si{\second}},
    ymajorgrids=true,
    grid style=dashed,
    width=1.0\columnwidth,
    height=4.5cm,
    legend pos=north west,
    legend columns=-1,
    x tick label style={/pgf/number format/1000 sep=},
%    ybar,
%    bar width=2pt,
    y filter/.code={\pgfmathparse{#1/1024}\pgfmathresult},
    only marks,
    each nth point=5,
    ymin=0,
]

\addplot
table[x index=0, y index=1]
{data/access_time.dat};
\addlegendentry{$N=2^{8}$}

\addplot
table[x index=0, y index=2]
{data/access_time.dat};
\addlegendentry{$N=2^{12}$}

\addplot
table[x index=0, y index=3]
{data/access_time.dat};
\addlegendentry{$N=2^{16}$}

\end{axis}
\end{tikzpicture}

\caption{Computation time for generating timestamps and commitments per database access for database size $N\in\{2^8,2^{12},2^{16}\}$ and block size $L=100\si{\kilo\byte}$.}
\label{fig:accesstime}
\end{figure}

In Figure~\ref{fig:accesstraffic} we show the communication cost for the client per database access for \textsf{ORAM-SS}, which is a combination of ORAM and Secret Sharing without long-term integrity protection, and \PROP.
For $N=2^{12}$ and $L=100\si{\kilo\byte}$ we observe that the communication overhead for adding long-term integrity protection ranges between $0.06\%$ in year 2018 and $259\%$ in 2118. The reason is that over time evidence is accumulated which needs to be communicated on every access.
The communication cost per database access scales logarithmically with the database size $N$ and linearly with the number of protection renewals. It also depends on the commitment and timestamp scheme parameters.

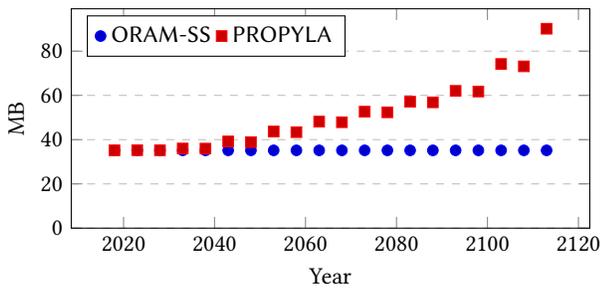
\begin{figure}
\centering

\begin{tikzpicture}
\begin{axis}[
    xlabel={Year},
    ylabel={\SI{}{\mega\byte}},
    ymajorgrids=true,
    grid style=dashed,
    width=1.0\columnwidth,
    height=4.5cm,
    legend pos=north west,
    legend columns=-1,
    x tick label style={/pgf/number format/1000 sep=},
%    ybar,
%    bar width=2pt,
    y filter/.code={\pgfmathparse{#1/1024/1024}\pgfmathresult},
    only marks,
    each nth point=5,
    ymin=0,
]

%\addplot
%table[x index=0, y index=1]
%{data/access_traffic.dat};
%\addlegendentry{$N=10^2$}

\addplot
table[x index=0, y index=4]
{data/access_traffic.dat};
\addlegendentry{\textsf{ORAM-SS}}

\addplot
table[x index=0, y index=3]
{data/access_traffic.dat};
\addlegendentry{\PROP}

%\addplot
%table[x index=0, y index=3]
%{data/access_traffic.dat};
%\addlegendentry{$N=10^4$}

\end{axis}
\end{tikzpicture}

\caption{Communication cost for the client per database access for database size $N = 2^{12}$ and block size $L=100\si{\kilo\byte}$.}
\label{fig:accesstraffic}
\end{figure}

\section{Conclusions}
We presented the first long-term secure storage architecture that combines long-term integrity, long-term confidentiality, and long-term access pattern hiding protection.
Overall, our performance measurements show that these protection goals can be achieved with manageable communication, computation, and storage costs.
The storage and communication cost per block and per storage server is independent of the block size $L$.
The computation and communication costs are logarithmic in the database size $N$ and also depend on the amount of accumulated evidence.
The protection renewal time is linear in $N$.

%We observe that for our instantiation, the most time consuming operation is the timestamp generation and the most storage consuming evidence data are the timestamps.
We envision that the performance of \PROP can further be improved by using Merkle Hash Trees (MHT) \cite{merkle1990certified} in order to reduce the number of timestamps as follows.
During timestamp and commitment renewal, instead of timestamping each data block separately, one could first create a MHT for the entire database and then timestamp only the root of that tree.
While asymptotically this increases the storage and computation complexity from $O(N)$ to $O(N\log(N))$, we expect that concrete computation costs decrease because typical hash functions are much faster to evaluate than the signing algorithms used for timestamping. Furhtermore, hash values are also much smaller in size compared to signatures.
We leave the exact description, implementation, and evaluation of this approach for future work.

%END CONTENT

%BEGIN BIB

\begin{acks}
This work has been co-funded by the \mbox{DFG} as part of projects S5 and S6 within the \mbox{CRC}~1119 \mbox{CROSSING}.
\end{acks}

\section*{Version history}
The following changes have been made since the proceedings version:
\begin{itemize}
	\item November 2018: An error in the description of the commitment renewal procedure in Algorithm~\ref{alg.access} has been fixed.
	The analysis of the communication cost in Section~\ref{sec.evaluation} has been revised and now shows the total communication cost for the client.
	\item April 2019: The number of shareholders was not specified in the evaluation and has been added.
\end{itemize}
    
\bibliographystyle{styles/ACM-Reference-Format}
%\balance
\bibliography{bib/bibliography}

%END BIB

\end{document}